%% file: soda.tex
\documentclass[11pt,letterpaper]{article}
\usepackage[letterpaper,margin=1in]{geometry}

\usepackage{cite}
\usepackage{amsmath}
\usepackage{amssymb}
\usepackage{amsfonts}
\usepackage{tabularx}
\usepackage{booktabs}
\usepackage{graphicx}
\usepackage{textcomp}
\usepackage{xcolor}
\usepackage{bm}

\usepackage{hyperref}
\hypersetup{
  colorlinks,
  linkcolor={green!60!black},
  citecolor={red!70!black},
  urlcolor={blue!70!black}
}

\usepackage{cleveref}
\usepackage{caption}
\usepackage{subcaption}

\usepackage{xurl}
\usepackage{makecell}
\input{def}

\newcommand{\supp}[1]{\textnormal{supp}\left(#1\right)}

\begin{document}

 \title{From Linear Exact Repair to Mixed-Alphabet Codes: The Bandwidth-Subpacketization Tradeoff of MDS Array Codes}
\author{Zihao Zhang \and Guodong Li \and Sihuang Hu}
\date{}
\maketitle

\begin{abstract}
MDS array codes achieve optimal storage efficiency in distributed storage systems, but their repair bandwidth depends critically on the sub-packetization level. While the theoretically optimal repair bandwidth (i.e., cut-set bound) of a high-rate exact-repair MDS array code is attainable with sufficiently large sub-packetization, the fundamental limits for fixed sub-packetization remain poorly understood. We study this problem for centralized linear exact repair of $h$ failures from $d$ helper nodes over a fixed finite field.

We show that the rank profile of any linear repair scheme is constrained by two auxiliary mixed-alphabet codes. The first code is formed from the kernels of the helper projections and exists for all repair degrees; the second is formed from their row spaces and applies in the full-helper regime when $d\ge r$. These codes have complementary alphabet profiles, and their size and minimum-distance requirements are forced solely by the MDS property and repair correctness. This gives a unified route from classical coding-theoretic upper bounds to repair-bandwidth lower bounds.
As a consequence, we obtain several explicit lower bounds via Singleton, Hamming, Plotkin, and Griesmer arguments, including known bounds as corollaries. We further formulate a Delsarte-type support-enumerator LP for mixed-alphabet codes and give an exact type-compressed implementation. For fixed sub-packetization \(\ell\), this yields a polynomial-time
procedure for evaluating the resulting finite-parameter lower bound.
\end{abstract}

\section{Introduction}
\label{sec:intro}

Maximum distance separable (MDS) array codes are a standard model for erasure-coded distributed storage. They achieve the optimal storage-reliability tradeoff: in an $(n,k,\ell)$ MDS array code over $\Fq$, each node stores a vector of length $\ell$, and the contents of any $k$ nodes determine the entire data object. The parameter $\ell$, called the sub-packetization, provides additional flexibility for repair. By splitting each node into $\ell$ symbols, a repair scheme can ask helper nodes to transmit carefully chosen linear combinations of their stored symbols, rather than downloading entire nodes. This flexibility can reduce repair bandwidth while preserving the MDS storage-reliability guarantee.

A central question is to understand the tradeoff between repair bandwidth and sub-packetization. The authors in~\cite{dimakis10} characterize the optimal bandwidth (i.e., cut-set bound) of MDS array codes when both sub-packetization and field size are sufficiently large. Since then, many explicit MSR and near-MSR constructions have been proposed, including product-matrix (PM) codes~\cite{rashmi11}, PM-RBT codes~\cite{rashmi15}, Butterfly codes~\cite{pamiesJuarez16}, Ye-Barg codes~\cite{ye17,ye17a}, generic transformations~\cite{li18}, Clay codes~\cite{vajha18}, and more recent universal MSR constructions~\cite{li24}; see~\cite{ramkumar22} for a survey. However, known lower bounds~\cite{alrabiah19,balaji18,balaji22} and high-rate constructions show that achieving the cut-set bound for high-rate MDS array codes under exact repair typically requires large, often exponential, sub-packetization.

In practice, a large $\ell$ increases the number of fragments that must be indexed, accessed, and combined during repair, leading to additional I/O and implementation overhead. This motivates the study of lower bounds for repair in the fixed sub-packetization regime.

Several constructions operate with small or constrained sub-packetization levels, including the piggybacking framework~\cite{rashmi17}, Hitchhiker codes~\cite{rashmi14}, HashTag codes~\cite{kralevska18-TBG, kralevska18-DASC}, elastic transformations~\cite{tang23}, $\varepsilon$-MSR codes~\cite{rawat17,ramkumar25}, and other related constructions~\cite{li21,li23MDS,cheng26,kadekodi23}. These works give upper bounds on the achievable repair bandwidth. In this paper, we focus on lower bounds for fixed $\ell$ and fixed field $\Fq$.

The authors of the survey~\cite{ramkumar22} formalized the general tradeoff question as follows.

\begin{openproblem}\cite[Open Problem 9]{ramkumar22}
\label{op:ramkumar}
Characterize the tradeoff between repair bandwidth, sub-packetization level, and field size for the general class of vector MDS codes.
\end{openproblem}

Lower bounds for fixed sub-packetization have been studied in several settings. For Reed--Solomon codes viewed over extension fields, Guruswami and Wootters~\cite{guruswami16} developed linear repair schemes and lower bounds using the subfield structure, and Dau and Milenkovic~\cite{dau17} refined these results for several families of full-length Reed--Solomon codes. For MDS array codes, Zhang, Li, and Hu~\cite{zhang2026} determined the average-min and max-min repair bandwidth for the special case $(k+2,k,\ell=2)$. More recently, Liu and Wu~\cite{Wu26-1} obtained a general projective-counting lower bound for linear exact repair in MDS array codes, and Wu~\cite{Wu26-2} refined this approach via an incidence-multiplicity argument. These works provide finite-parameter lower bounds from geometric or counting arguments.

In this paper, we study a local version of the fixed sub-packetization lower-bound problem. A local repair instance consists of a failure set $\mathcal H$ of size $h$ and a helper set $\mathcal D$ of size $d$. We ask for the minimum bandwidth of a centralized linear repair scheme for such an instance, minimized over all MDS array codes and over all choices of $(\mathcal H,\mathcal D)$ with these prescribed sizes. This local quantity is weaker than a global repair objective, which requires a single code to support low-bandwidth repair simultaneously over a prescribed family of repair instances. Nevertheless, it is a basic building block for understanding the full tradeoff.

Our approach is to recast each local repair scheme as a feasibility problem for auxiliary mixed-alphabet codes. This viewpoint turns repair-bandwidth lower bounds into code-size feasibility constraints and makes Delsarte-type LP methods available.

\subsection{Our contributions}

We study the local minimum bandwidth for the centralized linear repair of $h$ failed nodes from exactly $d$ helper nodes in an $(n,k,\ell)$ MDS array code over $\Fq$. A projection argument shows that it suffices to prove lower bounds for the reduced full-helper instance with $n'=h+d$ and redundancy $r'=h+d-k$. In the statements below, we use $r=n-k$ to denote this reduced redundancy.

\paragraph{Result 1: auxiliary code reduction.}
Every centralized linear repair scheme with bandwidth $\beta$ induces two auxiliary mixed-alphabet codes. Let
$
    \alpha=d\ell-\beta.
$
The null-space code $\cA$ exists for all repair degrees and satisfies
\[
    |\cA|\ge q^{\alpha-(r-h)\ell},
    \quad
    d_{\min}(\cA)\ge r+1.
\]
When $d\ge r$, the row-space code $\cB$ satisfies
\[
    |\cB|=q^{h\ell},
    \quad
    d_{\min}(\cB)\ge d-r+1.
\]
Hence, a feasible repair profile must pass mixed-alphabet code feasibility tests with prescribed coordinate dimensions, sizes, and distances.

\paragraph{Result 2: combinatorial profile constraints.} Applying classical coding-theoretic bounds to $\cA$ and $\cB$ translates code-size upper bounds into repair-bandwidth lower bounds. Singleton-type constraints recover the centralized cut-set bound. Hamming- and Plotkin-type constraints give additional finite-parameter profile inequalities. For example, Griesmer's bound applied to the flattened row-space code gives an explicit lower bound
\[
    \beta
    \ge
    \sum_{i=0}^{h\ell-1}
    \left\lceil \frac{d-r+1}{q^i}\right\rceil
\]
whenever $d\ge r$. Several known lower bounds on repair arise as special cases or relaxations of these profile constraints.

\paragraph{Result 3: exact type-compressed LP algorithm.} We formulate a Delsarte-type LP for mixed-alphabet codes with prescribed coordinate dimensions. The ordered support-enumerator LP has one variable and one dual constraint for each subset of the $d$ helper coordinates, and hence has size $2^d$. For fixed sub-packetization $\ell$, coordinates with the same local dimension are symmetric. Grouping supports by their type gives an exact type-compressed LP with at most $(d+1)^\ell$ support-type variables and constraints.

The resulting LP lower bound can be evaluated in time \(O_\ell\!\left(d^{\ell-1} T_{\mathrm{LP}}\bigl((d+1)^\ell\bigr)\right)\),
where $T_{\mathrm{LP}}(N)$ denotes the time to solve an LP with $O(N)$ variables and $O(N)$ constraints. Thus, for fixed \(\ell\), the mixed-alphabet feasibility test induced by linear repair can be evaluated in polynomial time in the repair degree.

\subsection{Technical overview}
The reduction starts from the matrices $MH_j$ induced by the repair matrix $M$. For each helper $j$, the kernel $ \bA_j=\ker(MH_j)$ records the local directions that are invisible to the repair matrix, while the row space $ \bB_j=\operatorname{row}(MH_j) $ records the downloaded linear combinations.

The null-space auxiliary code consists of all tuples
\[
    (u_j:j\in\mathcal D),\quad u_j\in \bA_j,
\]
whose contributions cancel in the parity-check equations:
\[
    \sum_{j\in\mathcal D} H_j u_j=0.
\]
A nonzero word supported on at most $r$ helpers would give a linear dependence among at most $r$ parity-check blocks, contradicting the MDS property. This gives a lower bound on the distance for $\cA$.

The row-space auxiliary code is the image
\[
    \cB=\{(xMH_j:j\in\mathcal D):x\in\Fq^{1\times h\ell}\}.
\]
When $d\ge r$, the MDS property implies that this map is injective and that every nonzero image has many nonzero block coordinates. Thus $\cB$ also has a prescribed size and minimum block distance.

We then apply classical coding bounds and, finally, a Delsarte LP to the resulting mixed-alphabet feasibility problem.

The algorithmic part comes from symmetry. The support-enumerator LP has one variable per support subset, but for fixed $\ell$, many coordinates have the same local dimension. Averaging over permutations of equal-type coordinates gives an equivalent LP whose variables are indexed only by support types.

\subsection{Organization}

Section~\ref{sec:pre} defines the centralized repair model and proves the full-helper reduction. Section~\ref{sec:reduc} gives the auxiliary-code reduction. Section~\ref{sec:bound} derives combinatorial profile bounds. Section~\ref{sec:LP} gives the type-compressed LP algorithm. We conclude in Section~\ref{sec:conclusion}.

\section{Preliminaries}\label{sec:pre}

Denote $[n]:=\{1,\dots,n\}$. We consider an $(n, k, \ell)$ linear MDS array code $\cC$ over the finite field $\Fq$, where the redundancy $r = n - k$. Each node $C_i$ stores a data vector of length $\ell$ over $\Fq$. %

The MDS property ensures that the contents of any set of at most $r=n-k$ erased nodes can be reconstructed from the remaining nodes; equivalently, any $k$ nodes determine the entire codeword. The structural relationship among the $n$ nodes is governed by the parity-check equations:
\begin{equation}\label{eq:pc}
    H_1C_1 + H_2C_2 + \cdots + H_nC_n = \bm{0},
\end{equation}
where each $H_i$ ($i\in[n]$) is an $r\ell \times \ell$ parity-check submatrix over $\Fq$. The MDS property holds if and only if any $r$ matrices chosen from $\{H_i: i\in[n]\}$ form an invertible $r\ell \times r\ell$ matrix when concatenated horizontally.

Without loss of generality, a linear repair scheme of a failure set $\cH \subseteq [n]$ consisting of $h$ ($1\le h\le r$) failed nodes can be described by a full-rank $h\ell \times r\ell$ \emph{repair matrix} $M$ over $\Fq$. Applying $M$ to the parity-check equations~\eqref{eq:pc} yields:
\begin{equation}
    \label{eq:repair}
    M H_1 C_1 + M H_2 C_2 + \cdots + M H_n C_n = \bm{0}.
\end{equation}

Assume $\cH=\{i_1,\dots,i_h\}$. Let $H_{\cH} = [H_{i_1} \dots H_{i_h}]$ be the $r\ell \times h\ell$ block matrix corresponding to the failed nodes. If the $h\ell \times h\ell$ square matrix $M H_{\cH}$ is invertible (i.e., has full rank $h\ell$), the failed nodes can be successfully recovered from the surviving nodes:
\[
    \begin{bmatrix} C_{i_1} \\ \vdots \\ C_{i_h} \end{bmatrix}= - \left(M H_{\cH} \right)^{-1} \sum_{j\in[n]\setminus \cH} M H_j C_j.
\]

During the repair process driven by $M$, the repair center needs the quantity $MH_jC_j$ from each surviving node $j$. Since the row space of $MH_j$ has dimension $\rk{MH_j}$, node $j$ only needs to transmit $\rk{MH_j}$ independent linear combinations of its $\ell$ stored symbols from which $MH_jC_j$ can be recovered.

Since $M H_{\cH}$ has full rank $h\ell$, each of its $h$ submatrices $M H_j$ ($j \in \cH$) must possess full column rank $\ell$. Therefore, the \emph{repair bandwidth} $\tB(M)$, defined as the total number of symbols transmitted over the network, is given by:
\begin{equation}\label{eq:BW}
    \tB(M) = \sum_{j\in[n]\setminus\cH} \rk{MH_j} = \sum_{j\in[n]} \rk{MH_j} - h\ell.
\end{equation}

For a given repair matrix $M$, we define the \emph{helper node set} $\cD$ as the subset of surviving nodes that actively participate in the repair process, i.e., $\cD = \{j\in[n]\setminus\cH: \rk{MH_j} > 0 \}$. The corresponding \emph{repair degree} $d$, representing the number of active helper nodes, is then $d = |\cD|$ and satisfies $k \le d \le n-h$.

Given a specific code $\cC$, a failure set $\cH$, and a designated set of helper nodes $\cD \subseteq [n]\setminus\cH$ with $|\cD| = d$, we are interested in valid repair matrices that require data downloads exclusively from $\cD$. Let $\cV(\cC, \cH, \cD)$ denote the subset of such valid repair matrices. Algebraically, any $M \in \cV(\cC, \cH, \cD)$ must satisfy the following conditions:
\begin{enumerate}
    \item $\rk{M H_{\cH}} = h\ell$;
    \item $\rk{M H_j} >0$ for all $ j\in\cD$;
    \item $ \rk{M H_j} =0$ for all $j\in[n]\setminus(\cH\cup\cD)$.
\end{enumerate}

We define $\beta(n,k,\ell,\Fq,h,d)$ as the minimum repair bandwidth required to repair $h$ failed nodes using exactly $d$ helper nodes under linear exact repair of $(n,k,\ell)$ MDS array codes over $\Fq$. Let $\mathfrak{C}$ denote the collection of all such MDS array codes. The fundamental limit $\beta$ is formally expressed as:
\begin{equation*}
    \beta(n,k,\ell,\Fq,h,d) := \min_{\cC\in\mathfrak{C}} \min_{\substack{\cH \subseteq [n] \\ |\cH| = h}} \min_{\substack{\cD \subseteq [n]\setminus\cH \\ |\cD| = d}} \min_{M \in \cV(\cC, \cH, \cD)} \tB(M).
\end{equation*}
If no $(n,k,\ell)$ MDS array code over $\mathbb F_q$ exists, we set $\beta(n,k,\ell,\mathbb F_q,h,d)=+\infty$.

We emphasize that $\beta$ is minimized over all codes, failure sets, helper sets, and valid repair matrices. Hence, any lower bound on $\beta$ applies universally to every linear repair instance with the prescribed parameters.

\begin{lemma}\label{lem:reduction_d}
The local-min repair bandwidth satisfies $$\beta(n,k,\ell,\Fq,h,d) \geq \beta(h+d, k, \ell,\Fq, h, d).$$
\end{lemma}

\begin{proof}
Suppose $d < n - h$. Let $\cH$ be the set of $h$ failed nodes, $\cD$ be the set of $d$ helper nodes, and $\cU = [n] \setminus (\cH \cup \cD)$ be the set of non-helper nodes, where $|\cU| = n -  d-h \le r-h$.

Choose $\cC,\cH,\cD$, and $M$ attaining $\beta(n,k,\ell,\Fq,h,d)$. Let $H = [H_1 \mid H_2 \mid \dots \mid H_n]$ be the parity-check matrix of $\cC$, where each $H_j \in \Fq^{r\ell \times \ell}$.

For each non-helper node $j \in \cU$, the strict requirement $MH_j = \mathbf{0}$ implies $M H_{\cU} = \mathbf{0}$, where $H_{\cU} \in \Fq^{r\ell \times |\cU|\ell}$ is the block matrix formed by concatenating $H_j$ for all $j \in \cU$. Since $|\cU| \le r$ and $\cC$ is MDS, $H_{\cU}$ has full column rank.

Let $r' = r - |\cU| = h + d + r - n$. Since $H_{\cU}$ has full column rank, choose a full-rank matrix $L\in \Fq^{r'\ell \times r\ell}$ whose row space is the left kernel of $H_{\cU}$. Then $L H_{\cU} = \mathbf{0}$, and the right null space of $L$ coincides with the column space of $H_{\cU}$. We construct a reduced array code $\cC'$ with parameters $(h+d,k,\ell)$, and new parity-check submatrices defined as $H'_j = L H_j \in \Fq^{r'\ell \times \ell}$ for $j\in \cD \cup \cH$.

We first verify that $\cC'$ preserves the MDS property. Let $S \subseteq \cD \cup \cH$ be any subset of size $|S| = r'$, and let $H_S$ denote the concatenation of blocks $H_j$ for $j \in S$. The concatenated matrix $[H_S \mid H_{\cU}]$ contains $|S| + |\cU| = r' + (n - h - d) = r$ blocks, meaning it is an invertible $r\ell \times r\ell$ matrix. Suppose there exists a vector $\mathbf{x}$ such that $L H_S \mathbf{x} = \mathbf{0}$. This implies that the vector $H_S\mathbf{x}$ lies in the kernel of $L$, which coincides with the column space of $H_{\cU}$. Thus, there exists a vector $\mathbf{y}$ such that $H_S \mathbf{x} = H_{\cU} \mathbf{y}$, or equivalently, $[H_S \mid H_{\cU}] [\mathbf{x}^T, -\mathbf{y}^T]^T = \mathbf{0}$. Since $[H_S \mid H_{\cU}]$ is full rank, it necessitates $\mathbf{x} = \mathbf{0}$. Therefore, $H'_S = L H_S$ is full rank, confirming that $\cC'$ is an $(n',k= n'-r', \ell)$ MDS array code.

Finally, since $M H_{\cU} = \mathbf{0}$, the rows of $M$ reside in the left null space of $H_{\cU}$, which is spanned by the rows of $L$. Hence, there exists a matrix $M' \in \Fq^{h\ell \times r'\ell}$ such that $M = M' L$. Since $M'H'_{\cH}=MH_{\cH}$, $M'$ is a valid repair matrix for $\cC'$. For any helper node $j \in \cD$, the downloaded bandwidth in the reduced system is
\[
\rk{M' H'_j} = \rk{M' L H_j} = \rk{M H_j}.
\]
This demonstrates that the exact repair bandwidth achieved in $\cC'$ using $M'$ equals the original optimal bandwidth. Since the theoretical minimum bandwidth of the reduced system cannot exceed this specific achievable value, it follows that $\beta(n, k, \ell, \Fq, h, d) \geq \beta(n', k, \ell, \Fq, h, d)$.
\end{proof}

Lemma~\ref{lem:reduction_d} implies that, for establishing lower bounds, it suffices to focus on the scenario in which all helpers participate. Equivalently, after substituting $n$ with $h+d$, we can assume that every surviving node participates in the repair process, which implies $d = n - h$.

\section{Reduction to Mixed-Alphabet Codes}
\label{sec:reduc}

We now reduce a fixed linear repair scheme to feasibility constraints for mixed-alphabet codes. Throughout this section, we work in the full-helper case. Thus, a set $\cH\subseteq[n]$ of $h$ nodes fails, all surviving nodes participate in the repair, and
\[
    \cD=[n]\setminus \cH,\quad d=|\cD|=n-h.
\]
Fix an $(n,k,\ell)$ MDS array code $\cC$ over $\Fq$, a failure set $\cH$, and a valid repair matrix $M\in \cV(\cC,\cH,\cD)$. Let
\[
    H=[H_1\mid H_2\mid \cdots \mid H_n]
\]
be a parity-check matrix of $\cC$, where each $H_j\in \Fq^{r\ell\times \ell}$ and $r=n-k$.

Unless otherwise stated, the weight of a mixed-alphabet codeword means its block Hamming weight, namely the number of nonzero block coordinates.

\begin{definition}[Local repair profile]
\label{def:local-profile}
Fix a valid repair matrix $M\in \cV(\cC,\cH,\cD)$. For each helper node $j\in\cD$, define
$$
    \bA_j:=\ker(MH_j) = \{u\in\Fq^{\ell\times 1}:MH_ju=\bm 0\},
$$
and
$$
    \bB_j:=\operatorname{row}(MH_j) =  \{vMH_j:v\in\Fq^{1\times h\ell}\}.
$$
Let $a_j:=\dim_{\Fq}\bA_j$ and $b_j:=\dim_{\Fq}\bB_j.$
We call the two sequences $(a_j:j\in\cD)$ and $(b_j:j\in\cD)$ the local repair profiles induced by $M$.
\end{definition}

By rank-nullity, $a_j+b_j=\ell,\quad j\in\cD.$ Set $\alpha:=\sum_{j\in\cD}a_j.$ If $\beta=\tB(M)$ is the bandwidth of the repair scheme, then
\begin{equation}\label{eq:bw-relation}
    \beta
    = \sum_{j\in\cD}b_j
    = d\ell-\alpha .
\end{equation}

The spaces $\bA_j$ and $\bB_j$ encode complementary information: $\bA_j$ records the part of node $j$ not accessed by the repair matrix, while $\bB_j$ records the row space actually downloaded from node $j$. We next turn these local spaces into two auxiliary codes.

\subsection{The null-space auxiliary code}

\begin{definition}[Null-space code]
\label{def:code_A}
Let $\mathbb A:=\prod_{j\in\cD}\bA_j,$ then we define
\[
    \cA:=
    \left\{
        u=(u_j:j\in\cD)\in\mathbb A:
        \sum_{j\in\cD}H_ju_j=\bm 0
    \right\}.
\]
\end{definition}

\begin{lemma}[Properties of the null-space code]
\label{lem:code_A}
The code $\cA$ is an $\Fq$-linear mixed-alphabet code with coordinate alphabets $(\bA_j:j\in\cD)$. Moreover,
\[
    |\cA|\ge q^{\alpha-(r-h)\ell},
    \quad
    d_{\min}(\cA)\ge r+1.
\]
\end{lemma}

\begin{proof}
Define $\Phi:\mathbb A\to\Fq^{r\ell}$, and $\Phi(u)=\sum_{j\in\cD}H_ju_j.$
Then $\cA=\ker\Phi$. For every $u\in\mathbb A$,
\[
    M\Phi(u)=\sum_{j\in\cD}MH_ju_j=\bm 0,
\]
because $u_j\in\ker(MH_j)$. Hence
$
    \operatorname{Im}\Phi\subseteq \ker M.
$
Since $M$ has full row rank $h\ell$,
$
    \dim_{\Fq}(\ker M)=(r-h)\ell.
$
Therefore
\[
    \dim_{\Fq}\cA
    =
    \dim_{\Fq}\mathbb A-\dim_{\Fq}\operatorname{Im}\Phi
    \ge
    \alpha-(r-h)\ell.
\]

It remains to prove the distance bound. Suppose that $0\ne u\in\cA$ has support
$
    J=\{j\in\cD:u_j\ne\bm 0\}
$
with $|J|\le r$. Then
$
    \sum_{j\in J}H_ju_j=\bm 0.
$
Let $H_J$ be the concatenation of the blocks $H_j$ for $j\in J$, and let $u_J$ be the corresponding concatenated vector. Then
$
    H_Ju_J=\bm 0.
$
By the MDS property, every collection of at most $r$ parity-check blocks has full column rank. Hence $u_J=\bm 0$, contradicting the definition of $J$. Thus every nonzero codeword of $\cA$ has block Hamming weight at least $r+1$.
\end{proof}

The construction of $\cA$ does not require any relation between $d$ and $r$. In particular, if $d\le r$, then $d_{\min}(\cA)\ge r+1>d$, so $\cA=\{\bm 0\}$.

\subsection{The row-space auxiliary code}

The second auxiliary code uses the row spaces $(\bB_j:j\in\cD)$. Its strong size and distance guarantees require $d\ge r$.

\begin{definition}[Row-space code]
\label{def:code_B}
Let $ \mathbb B:=\prod_{j\in\cD}\bB_j.$
Define $\Phi_{\cB}:\Fq^{1\times h\ell}\to \mathbb B$ and $\Phi_{\cB}(x)=(xMH_j:j\in\cD).$
The row-space auxiliary code is
\[
    \cB:=\operatorname{Im}\Phi_{\cB}
    =
    \left\{
        (xMH_j:j\in\cD):x\in\Fq^{1\times h\ell}
    \right\}.
\]
\end{definition}

By definition, $\cB$ is an $\Fq$-linear mixed-alphabet code of length $d$ with coordinate alphabets $(\bB_j:j\in\cD)$. In general, $|\cB|\leq q^{h\ell}$.

\begin{lemma}[Properties of the row-space code]
\label{lem:code_B}
If $d\ge r$, then $|\cB|=q^{h\ell}$ and $d_{\min}(\cB)\ge d-r+1.$
\end{lemma}

\begin{proof}
We first prove that $\Phi_{\cB}$ is injective. Suppose $\Phi_{\cB}(x)=\bm 0$. Then $xMH_j=\bm 0, j\in\cD.$ Since $d\ge r$, choose $R\subseteq\cD$ with $|R|=r$. Let $H_R$ be the concatenation of the blocks $H_j$ for $j\in R$. Then $xMH_R=\bm 0. $ By the MDS property, $H_R$ is invertible. Hence $xM=\bm 0$. Since $M$ has full row rank $h\ell$, this forces $x=\bm 0$. Thus $\Phi_{\cB}$ is injective, and $|\cB|=q^{h\ell}.$

Now let $0\ne v_x=\Phi_{\cB}(x)\in\cB$, and set $y=xM\in\Fq^{1\times r\ell}.$ Since $M$ has full row rank, $y \neq \bm 0$. Let $Z=\{j\in\cD:yH_j=\bm 0\}$ be the set of zero block coordinates of $v_x$. If $|Z|\ge r$, choose $R\subseteq Z$ with $|R|=r$. Then $yH_R=\bm 0.$ Since $H_R$ is invertible, this implies $y=\bm 0$, a contradiction. Therefore $|Z|\le r-1$, and $\operatorname{wt}(v_x)=d-|Z|\ge d-r+1.$
\end{proof}

We summarize the reduction in the following theorem.

\begin{theorem}[Auxiliary-code reduction]
\label{thm:main-reduction}
Consider a centralized linear repair scheme with repair matrix $M$ and bandwidth $\beta$. Let $\alpha:=d\ell-\beta$. Then the scheme induces a null-space auxiliary code $\cA$ in Definition~\ref{def:code_A} satisfying
\[
    |\cA|\ge q^{\alpha-(r-h)\ell},
    \quad
    d_{\min}(\cA)\ge r+1 .
\]
If $d\ge r$, the same scheme further induces a row-space auxiliary code $\cB$ in Definition~\ref{def:code_B} satisfying
\[
    |\cB|=q^{h\ell},
    \quad
    d_{\min}(\cB)\ge d-r+1 .
\]
\end{theorem}
\begin{proof}
The identity $\alpha=d\ell-\beta$ follows from Definition~\ref{def:local-profile} and \eqref{eq:bw-relation}. The properties of $\cA$ follow from Lemma~\ref{lem:code_A}. If $d\ge r$, the properties of $\cB$ follow from Lemma~\ref{lem:code_B}.
\end{proof}

\begin{remark}
The condition $d\ge r$ in Lemma~\ref{lem:code_B} is necessary for the argument above, as the proof selects $r$ helper coordinates. In the common high-rate setting $r \le k$, this condition is satisfied automatically because $d \ge k$.
\end{remark}

\section{Combinatorial Profile Bounds}\label{sec:bound}

The reduction in Section~\ref{sec:reduc} converts every feasible repair scheme into auxiliary mixed-alphabet codes with prescribed alphabet profiles. Therefore, any upper bound on mixed-alphabet codes with a given distance becomes a necessary condition on feasible bandwidth profiles. This section records the resulting constraints from classical coding bounds. Table~\ref{tab:profile-known-bounds} summarizes how the constraints derived in this section recover several existing repair lower bounds as special cases or relaxations.

\begin{table}[ht]
\centering
\small
\caption{Classical coding bounds as profile constraints and their consequences.}
\label{tab:profile-known-bounds}
\begin{tabularx}{\linewidth}{@{}p{0.16\linewidth}p{0.30\linewidth}X@{}}
\toprule
Coding bound & Auxiliary-code constraint & Consequences \\
\midrule

Singleton
&
Subset-sum constraint for the null-space code $\cA$;
Lemma~\ref{lem:singleton}.
&
Summing over all $r$-subsets recovers the centralized cut-set bound
$\beta\ge dh\ell/r$~\cite{dimakis10,Centra-cut-set2018};
see Remark~\ref{rem:cut-set}.
\\
\addlinespace

Hamming
&
Sphere-packing constraints for $\cA$, and for $\cB$ when $d\ge r$;
Lemma~\ref{lem:hamming}.
&
The radius-one relaxation of the $\cA$-side constraint recovers the
projective-counting bound of~\cite{Wu26-1};
see Remark~\ref{rem:wu26-1}.
\\
\addlinespace

Plotkin
&
Average-weight constraints for $\cA$, and for $\cB$ when $d\ge r$;
Lemma~\ref{lem:plotkin}.
&
The $\cB$-side constraint, rewritten using $a_j+b_j=\ell$ and combined with Jensen's inequality, yields the Guruswami--Wootters-type bound~\cite{guruswami16}, while a linear relaxation recovers the incidence-multiplicity bound of \cite{Wu26-2}; see Remark~\ref{rem:guruswami-wu}.
\\
\addlinespace

Griesmer
&
Scalar Griesmer bound applied to the flattened code $\cB$;
Corollary~\ref{coro:griesmer}.
&
Gives an explicit lower bound
$\beta\ge
    \sum_{i=0}^{h\ell-1}
    \left\lceil\frac{d-r+1}{q^i}\right\rceil$.
\\

\bottomrule
\end{tabularx}
\end{table}

\subsection{Basic Profile Constraints}
\label{subsec:profile-constraints}

We first record three families of constraints obtained from standard coding bounds. The Singleton-type constraint comes from puncturing, the Hamming-type constraint from sphere packing, and the Plotkin-type constraint from double counting nonzero coordinates.

\begin{lemma}[Singleton-Type Constraint]\label{lem:singleton}
Assume $d\ge r$. For every subset $J\subseteq\cD$ with $|J|=r$, the profile $(a_j)_{j\in\cD}$ satisfies
$
    \sum_{j\in J}a_j \le (r-h)\ell.
$
Equivalently,
$
    \sum_{j\in J}b_j \ge h\ell.
$
\end{lemma}

\begin{proof}
By Lemma~\ref{lem:code_A}, the code $\cA$ has minimum block distance at least $r+1$. Hence, the projection of $\cA$ onto any subset $S\subseteq\cD$ with $|S|=d-r$ is injective. Therefore
\[
    |\cA|\le \prod_{j\in S}q^{a_j}=q^{\sum_{j\in S}a_j}.
\]
Together with $|\cA|\ge q^{\alpha-(r-h)\ell}$, this gives
$
    \alpha-(r-h)\ell \le \sum_{j\in S}a_j.
$
Since $\alpha=\sum_{j\in\cD}a_j$, taking $J=\cD\setminus S$ gives
$
    \sum_{j\in J}a_j \le (r-h)\ell.
$
The equivalent formulation in terms of $(b_j)_{j\in\cD}$ follows from $a_j+b_j=\ell$.
\end{proof}

\begin{remark}\label{rem:cut-set}
Summing the inequality $\sum_{j\in J}b_j\ge h\ell$ over all $r$-subsets $J\subseteq\cD$ gives the centralized cut-set bound $\beta\ge \frac{dh\ell}{r}$ \cite{dimakis10,Centra-cut-set2018}.
\end{remark}

The next constraint strengthens the puncturing argument by using disjoint Hamming balls in the mixed-alphabet ambient spaces.

\begin{lemma}[Hamming-Type Constraint]\label{lem:hamming}
Let $e_A=\lfloor r/2\rfloor$, then
\[
    \sum_{i=0}^{e_A}
    \sum_{\substack{J\subseteq\cD:\ |J|=i}}
    \prod_{j\in J}(q^{a_j}-1)
    \le q^{(r-h)\ell}.
\]
Moreover, assuming $e_B=\lfloor(d-r)/2\rfloor\ge 0$, then
\[
    \sum_{i=0}^{e_B}
    \sum_{\substack{J\subseteq\cD:\ |J|=i}}
    \prod_{j\in J}(q^{b_j}-1)
    \le q^{\beta-h\ell}.
\]
\end{lemma}

\begin{proof}
The Hamming ball of radius $e_A$ in $\prod_{j\in\cD}\Fq^{a_j}$ has volume
\[
    V_A(e_A)=
    \sum_{i=0}^{e_A}
    \sum_{\substack{J\subseteq\cD:\ |J|=i}}
    \prod_{j\in J}(q^{a_j}-1).
\]
Since $d_{\min}(\cA)\ge r+1$, such balls around the codewords of $\cA$ are disjoint. Thus $|\cA|V_A(e_A)\le q^\alpha$. Using $|\cA|\ge q^{\alpha-(r-h)\ell}$ gives the first inequality. The second inequality follows in the same way from $|\cB|=q^{h\ell}$, $d_{\min}(\cB)\ge d-r+1$, and the ambient size $q^\beta$.
\end{proof}

\begin{remark}\label{rem:wu26-1}
For $h=1$ and $d=n-1$, keeping only the radius-one terms in the $\cA$-side Hamming constraint gives
\[
    1+\sum_{j\in\mathcal D}(q^{a_j}-1) \le q^{(r-1)\ell}.
\]
Equivalently,
\[
    \sum_{j\in\mathcal D}q^{a_j} \le q^{(r-1)\ell}+d-1.
\]
This recovers the projective-counting bound of~\cite{Wu26-1} in the present notation.
\end{remark}

We also need a Plotkin-type constraint. This form is useful when the distance requirement is large relative to the number of helper coordinates.

\begin{lemma}[Plotkin-Type Constraint]\label{lem:plotkin}
When $d\ge r$, the profile $(a_j)_{j\in\cD}$ satisfies
\[
    \sum_{j\in\cD} q^{a_j}
    \le
    (r-1)q^\ell+
    \frac{d-r+1}{q^{(h-1)\ell}}.
\]
The profile $(b_j)_{j\in\cD}$ satisfies
\[
    \sum_{j\in\cD} q^{b_j}
    \le
    (d-r-1)q^\ell+
    (r+1)q^{\beta-(k-1)\ell}.
\]
\end{lemma}

\begin{proof}
For $\cB$, let $S_B=\sum_{v\in\cB}\operatorname{wt}(v)$. Since $d_{\min}(\cB)\ge d-r+1$,
\[
    S_B\ge (q^{h\ell}-1)(d-r+1).
\]
On the other hand, the $j$-th coordinate projection of $\cB$ has dimension $b_j$, so
\[
    S_B=q^{h\ell}\sum_{j\in\cD}(1-q^{-b_j}).
\]
Combining these estimates gives
\[
    \sum_{j\in\cD}q^{-b_j} \le r-1+(d-r+1)q^{-h\ell}.
\]
Using $b_j=\ell-a_j$ gives the first inequality.

For $\cA$, let $c_j:=\dim \operatorname{proj}_j(\cA)\le a_j$. Then
\[
       \sum_{u\in\cA}\operatorname{wt}(u)
    = |\cA|\sum_{j\in\mathcal D}(1-q^{-c_j})
    \le |\cA|\sum_{j\in\mathcal D}(1-q^{-a_j}).
\]
Since $d_{\min}(\cA)\ge r+1$,
\[
    (|\cA|-1)(r+1) \le |\cA|\sum_{j\in\mathcal D}(1-q^{-a_j}).
\]
Since $|\cA|\ge q^{\alpha-(r-h)\ell}$, we obtain
\[
    \sum_{j\in\cD}q^{-a_j} \le d-r-1+(r+1)q^{(r-h)\ell-\alpha}.
\]
Substituting $a_j=\ell-b_j$, $\alpha=d\ell-\beta$, and $d=n-h$ gives the second inequality.
\end{proof}

\begin{remark}\label{rem:guruswami-wu}
For $h=1$ and $d=n-1$, the first Plotkin-type constraint recovers the lower bound in~\cite{guruswami16} via Jensen's inequality. The weaker linear relaxation $q^t\ge 1+(q-1)t$ for integers $t\ge 0$ also recovers the corresponding bound in~\cite{Wu26-2}.
\end{remark}

\subsection{Representative Lower Bounds on Repair Bandwidth}
\label{subsec:representative-lower-bounds}

We next translate the profile constraints above into explicit lower bounds on $\beta$. We state only representative closed-form consequences.

The first consequence follows from applying the Griesmer bound to the flattened row-space code $\cB$.

\begin{corollary}[Griesmer-Type Lower Bound]\label{coro:griesmer}
Assume $d\ge r$. Then every centralized linear repair scheme satisfies
\[
    \beta \ge
    \sum_{i=0}^{h\ell-1}
    \left\lceil \frac{d-r+1}{q^i}\right\rceil .
\]
\end{corollary}

\begin{proof}
Flatten each block $\Fq^{b_j}$ of $\cB$ into $b_j$ scalar coordinates over $\Fq$. The resulting scalar linear code has length $\sum_{j\in\cD} b_j=\beta$, dimension $h\ell$, and minimum scalar distance at least $d-r+1$. The standard Griesmer bound gives the result.
\end{proof}

The Hamming-type constraints give several simple closed-form consequences.

\begin{corollary}[Hamming-Type Lower Bound]\label{coro:hamming-average}
Let $e=\lfloor r/2\rfloor$. If $e\ge 1$, then every centralized linear repair scheme satisfies
\[
    \beta \ge d\ell- \frac{d}{e}\log_q\left(1+\frac{e}{d}\bigl(q^{(r-h)\ell}-1\bigr)\right).
\]
\end{corollary}

\begin{proof}
By Lemma~\ref{lem:hamming},
\[
    \sum_{s=1}^{e} \sum_{\substack{S\subseteq\cD:\ |S|=s}} \prod_{j\in S}(q^{a_j}-1) \le q^{(r-h)\ell}-1.
\]
For each $J\subseteq\cD$ with $|J|=e$,
\[
    q^{\sum_{j\in J}a_j}-1 =\sum_{\emptyset\ne S\subseteq J}\prod_{j\in S}(q^{a_j}-1).
\]
Summing over all such $J$ and using $\binom{d-s}{e-s}\le\binom{d-1}{e-1}$ gives
\[
        \sum_{\substack{J\subseteq\cD:\ |J|=e}} q^{\sum_{j\in J}a_j}
    \le \binom{d}{e}+\binom{d-1}{e-1}\bigl(q^{(r-h)\ell}-1\bigr).
\]
By AM-GM, the left-hand side is at least $\binom{d}{e}q^{e\alpha/d}$. Therefore
\[
    q^{e\alpha/d} \le 1+\frac{e}{d}\bigl(q^{(r-h)\ell}-1\bigr).
\]
Substituting $\alpha=d\ell-\beta$ gives the result.
\end{proof}

Finally, the first Plotkin-type constraint gives a clean Jensen-type lower bound.

\begin{corollary}[Plotkin-Jensen Lower Bound]\label{coro:plotkin-jensen}
Assume that $d\ge r$. Then every centralized linear repair scheme satisfies
\[
    \beta \ge
    dh\ell-d\log_q\left(
        1+\frac{(r-1)(q^{h\ell}-1)}{d}
    \right).
\]
\end{corollary}

\begin{proof}
By Jensen's inequality,
\[
    \sum_{j\in\cD}q^{a_j} \ge d q^{\alpha/d}.
\]
Combining this with the first Plotkin-type constraint in Lemma~\ref{lem:plotkin}, and then substituting $\alpha=d\ell-\beta$, gives the claimed lower bound.
\end{proof}

\begin{remark}
 Numerous results can be obtained from combinatorial constraints. Since those consequences are not uniformly valid as clean lower bounds and are usually better handled numerically, we do not state those as separate corollaries.
\end{remark}

The inequalities derived in this section are clear and easy to evaluate. For tighter bounds that are valid at finite parameters, we will use the type-compressed linear program presented in the following section.

\section{Linear-Programming Bounds}
\label{sec:LP}

We use Delsarte's linear programming method in a form adapted to mixed-alphabet codes. Given a coordinate-dimension profile, the LP upper-bounds the size of any mixed-alphabet code with prescribed block distance. Combined with the auxiliary codes from Section~\ref{sec:reduc}, this gives lower bounds on repair bandwidth.

Throughout this section, we work in the full-helper setting:
\[
    r=n-k,\quad
    \cD=[n]\setminus\cH,\quad
    d=|\cD|=n-h.
\]
Recall that
\[
    b_j=\dim\bB_j\in\{1,\ldots,\ell\},
    \quad
    a_j=\dim\bA_j=\ell-b_j\in\{0,\ldots,\ell-1\},
\]
and
\[
    \beta=\sum_{j\in\cD}b_j = d\ell-\sum_{j\in\cD}a_j.
\]
Thus, the bandwidth profile has no zero-dimensional coordinates, whereas the complementary null-space profile may have zero-dimensional coordinates.

\subsection{A mixed-alphabet LP method}
\label{subsec:LP-framework}

Let
\[
    \bV=\prod_{i\in\cD}\bV_i,
    \quad
    \dim_{\Fq}\bV_i=t_i\ge 1.
\]
Let $\cR\subseteq\bV$ be an $\Fq$-linear code with minimum block distance at least $w$.

For $v\in\bV$, write
$
    \supp{v}=\{i\in\cD:v_i\ne \bm 0\}.
$
For $S\subseteq\cD$, define
\[
    E_S=\bigl|\{v\in\cR:\supp{v}=S\}\bigr|.
\]
Then
\[
    E_\emptyset=1, \quad E_S=0\quad\text{for }0<|S|<w.
\]

Let $\widehat{\bV}_i$ be the character group of the additive group $\bV_i$, and let $\cE_i$ denote the trivial character. A character $\chi\in\widehat{\bV}:=\prod_{i\in\cD}\widehat{\bV}_i$ decomposes as $\chi=\prod_i\chi_i$, and its support is
\[
    \supp{\chi}=\{i\in\cD:\chi_i\ne\cE_i\}.
\]
The dual code is
\[
    \cR^\perp = \{\chi\in\widehat{\bV}:\chi(v)=1\text{ for all }v\in\cR\}.
\]
Let
$
    E_S^\perp = \bigl|\{\chi\in\cR^\perp:\supp{\chi}=S\}\bigr|.
$
For $T,S\subseteq\cD$, define
\[
    K_T(S) := \sum_{\substack{\chi\in\widehat{\bV}\\ \supp{\chi}=S}}\chi(v),
\]
where $v$ is any vector satisfying $\supp{v}=T$. Thus, the first argument of $K_T(S)$ is a codeword support, and the second argument is a character support. The value depends only on $T$ and $S$, not on the choice of $v$.

\begin{lemma}[Support kernel]
\label{lem:support-kernel}
For all $T,S\subseteq\cD$,
\[
    K_T(S) = \left(\prod_{i\in S\setminus T}(q^{t_i}-1)\right) (-1)^{|S\cap T|}.
\]
\end{lemma}
\begin{proof}
The character sum factors over coordinates. If $i\notin S$, the local character is trivial and contributes $1$. If $i\in S\setminus T$, then $v_i=0$, and the sum over nontrivial characters contributes $q^{t_i}-1$. If $i\in S\cap T$, then $v_i\ne0$, and character orthogonality gives
$
    \sum_{\chi_i\ne\cE_i}\chi_i(v_i)=-1.
$
Multiplying these local contributions over all coordinates gives the formula.
\end{proof}

\begin{lemma}[MacWilliams transform by supports]
\label{lem:Enum_perp}
For every $S\subseteq\cD$,
\[
    E_S^\perp
    =
    \frac{1}{|\cR|}
    \sum_{T\subseteq\cD}E_TK_T(S).
\]
\end{lemma}
\begin{proof}
Evaluate $
    \eta
    =
    \sum_{\substack{\chi\in\widehat{\bV}\\ \supp{\chi}=S}}
    \sum_{v\in\cR}\chi(v)
$ in two ways. If $\chi\in\cR^\perp$, the inner sum is $|\cR|$; otherwise it is $0$ by translation invariance over $\cR$. Hence
$
    \eta=E_S^\perp|\cR|.
$
Grouping the same sum by the support $T$ of $v\in\cR$ gives
$
    \eta=\sum_{T\subseteq\cD}E_TK_T(S).
$
The claim follows.
\end{proof}
\begin{theorem}[Parametrized LP bound]
\label{thm:L_Bound_parametrized}
Fix a local-dimension profile $\mathbf t=(t_i:i\in\cD)$. If a mixed-alphabet linear code $\cR$ with profile $\mathbf t$ and minimum block distance of at least $w$ exists, then $|\cR|$ is at most the optimum of the following linear program:
\[
\begin{array}{ll}
\textnormal{maximize}
&
\displaystyle \sum_{T\subseteq\cD}E_T
\\[2mm]
\textnormal{subject to}
&
E_\emptyset=1,
\\[1mm]
&
E_T=0,
\quad 1\le |T|<w,
\\[1mm]
&
E_T\ge0,
\quad |T|\ge w,
\\[1mm]
&
\displaystyle
\sum_{T\subseteq\cD}E_TK_T(S)\ge0,
\quad
\forall S\subseteq\cD .
\end{array}
\tag{LP}
\label{eq:ordered_LP}
\]
\end{theorem}

\begin{proof}
For an actual code, $E_T$ is its exact support enumerator, so the objective equals $|\cR|$. The first three constraints encode the zero word, the minimum-distance condition, and non-negativity. The last constraint is equivalent to $E_S^\perp\ge0$, by Lemma~\ref{lem:Enum_perp}. Therefore, every code yields a feasible LP solution with objective value $|\cR|$.
\end{proof}

\begin{remark}[Dual viewpoint]
\label{rem:dual-viewpoint}
The LP also has the standard Delsarte dual formulation. In the dual formulation, feasible dual solutions can be interpreted as test functions
on support sets. The Singleton-, Hamming-, and Plotkin-type inequalities in Section~\ref{sec:bound} can be recovered from suitable choices of such
test functions. We use this remark only to indicate the relation with the closed-form bounds; the rest of the paper works with the primal support-enumerator formulation.
\end{remark}
\subsection{Exact type compression}
\label{subsec:type-compression}

Let $U_q(\mathbf t,\omega)$ denote the optimum of \eqref{eq:ordered_LP} for local-dimension profile $\mathbf t$ and minimum block distance at least $\omega$. If $t_i=0$, then coordinate $i$ contains only the zero symbol in every codeword. Deleting such coordinates preserves both code size and block distance; after deletion, we relabel the remaining coordinate set as $\cD$. If the remaining length is smaller than $\omega$, only the zero code is possible, and we set
$
    U_q(\mathbf t,\omega)=1.
$

The LP in \eqref{eq:ordered_LP} has one variable per support subset. In the repair application, coordinates with equal local dimensions are symmetric. We exploit this symmetry exactly.

Let
$
    \mathbf m=(m_1,\ldots,m_L)
$
be a type profile, where $m_s$ is the number of coordinates of local dimension $s$, and $L$ is the largest local dimension appearing in the profile. Define
\[
    \mathcal X(\mathbf m)
    :=
    \left\{
    \mathbf x\in\mathbb Z_{\ge0}^{L}:
    0\le x_s\le m_s,\ s=1,\ldots,L
    \right\}.
\]
For $\mathbf x\in\mathcal X(\mathbf m)$, set
\[
    |\mathbf x|:=\sum_{s=1}^{L}x_s,
    \quad
    N(\mathbf x):=\prod_{s=1}^{L}\binom{m_s}{x_s}.
\]

Given $\mathbf m$, partition the coordinate set as
\[
    \cD=\cD_1\sqcup\cdots\sqcup\cD_L,
    \quad
    \cD_s=\{i\in\cD:t_i=s\},
    \quad
    |\cD_s|=m_s.
\]
For a subset $S\subseteq\cD$, define its type by
\[
    \operatorname{type}_{\mathbf m}(S)
    :=
    \bigl(|S\cap\cD_1|,\ldots,|S\cap\cD_L|\bigr)
    \in\mathcal X(\mathbf m).
\]
When $\mathbf m$ is clear, we write $\operatorname{type}(S)$.

For $\mathbf x,\mathbf y\in\mathcal X(\mathbf m)$, define
\[
    \mathsf K_{\mathbf x}(\mathbf y)
    :=
    \prod_{s=1}^{L}
    \left[
    \sum_{z=0}^{\min\{x_s,y_s\}}
    \binom{y_s}{z}
    \binom{m_s-y_s}{x_s-z}
    (-1)^z
    (q^s-1)^{y_s-z}
    \right],
\]
where $\binom ab=0$ if $b<0$ or $b>a$.

\begin{lemma}[Type kernel]
\label{lem:type-kernel}
Fix a type profile $\mathbf m$. Let $S\subseteq\cD$ have type
$
    \operatorname{type}(S)=\mathbf y.
$
Then, for every $\mathbf x\in\mathcal X(\mathbf m)$,
\[
    \sum_{\substack{T\subseteq\cD\\ \operatorname{type}(T)=\mathbf x}}
    K_T(S)
    =
    \mathsf K_{\mathbf x}(\mathbf y).
\]
\end{lemma}

\begin{proof}
Fix a dimension class $\cD_s$. Exactly $y_s$ coordinates of $\cD_s$ lie in $S$. If $T$ has type $\mathbf x$, then $T$ contains $x_s$ coordinates from $\cD_s$. Suppose $z$ of them lie in $S$. There are
\[
    \binom{y_s}{z}\binom{m_s-y_s}{x_s-z}
\]
ways to make this choice. The $z$ coordinates in $S\cap T\cap\cD_s$ contribute $(-1)^z$, while the $y_s-z$ coordinates in $(S\setminus T)\cap\cD_s$ contribute $(q^s-1)^{y_s-z}$. Coordinates in $T\setminus S$ contribute $1$. Thus, the contribution of the dimension-$s$ class is
\[
    \sum_{z=0}^{\min\{x_s,y_s\}}
    \binom{y_s}{z}
    \binom{m_s-y_s}{x_s-z}
    (-1)^z(q^s-1)^{y_s-z}.
\]
Multiplying over all $s$ gives the formula.
\end{proof}

Let $e_{\mathbf x}$ denote the common value of $E_T$ over all supports of type $\mathbf x$. Hence, the objective is weighted by
\[
    N(\mathbf x)=|\{T\subseteq\cD:\operatorname{type}(T)=\mathbf x\}|,
\]
the number of such supports.

Let $U_q^{\mathrm{typ}}(\mathbf m,\omega)$ be the optimum of the
following type-compressed LP:
\[
\begin{array}{ll}
\textnormal{maximize}
&
\displaystyle
\sum_{\mathbf x\in\mathcal X(\mathbf m)}
N(\mathbf x)e_{\mathbf x}
\\[2mm]
\textnormal{subject to}
&
e_{\mathbf 0}=1,
\\[1mm]
&
e_{\mathbf x}=0,
\quad 1\le|\mathbf x|<\omega,
\\[1mm]
&
e_{\mathbf x}\ge0,
\quad |\mathbf x|\ge\omega,
\\[1mm]
&
\displaystyle
\sum_{\mathbf x\in\mathcal X(\mathbf m)}
e_{\mathbf x}\mathsf K_{\mathbf x}(\mathbf y)
\ge0,
\quad
\forall \mathbf y\in\mathcal X(\mathbf m).
\end{array}
\tag{TC-LP}
\label{eq:type_compressed_LP}
\]

If $\sum_{s=1}^{L}m_s<\omega$, we set
$
    U_q^{\mathrm{typ}}(\mathbf m,\omega)=1.
$

\begin{proposition}[Exact type compression]
\label{prop:type-compression}
Let $\mathbf m(\mathbf t)$ be the type profile of $\mathbf t$, after
deleting zero-dimensional coordinates. Then
\[
    U_q(\mathbf t,\omega)
    =
    U_q^{\mathrm{typ}}(\mathbf m(\mathbf t),\omega).
\]
\end{proposition}

\begin{proof}
We write $\mathbf m=\mathbf m(\mathbf t)$, and let $G=\prod_{s=1}^{L}\operatorname{Sym}(m_s)$ be the group that permutes coordinates of equal local dimension.

First, take any feasible solution $(E_T: T\subseteq\cD)$ of the ordered LP and average it over $G$. The objective is unchanged, and feasibility is preserved because $G$ preserves support sizes and local dimensions. After averaging, $E_T$ depends only on $\operatorname{type}(T)$. Write this common value as $e_{\mathbf x}$ for supports of type $\mathbf x$. The objective becomes
$
    \sum_{\mathbf x\in\mathcal X(\mathbf m)}
    N(\mathbf x)e_{\mathbf x}.
$
The zero-word, distance, and non-negativity constraints become the corresponding constraints of $(\mathrm{TC\text{-}LP})$. For a support $S$ of type $\mathbf y$, the ordered dual constraint becomes
\[
    \sum_{\mathbf x\in\mathcal X(\mathbf m)}
    e_{\mathbf x}
    \sum_{\substack{T\subseteq\cD\\\operatorname{type}(T)=\mathbf x}}
    K_T(S)
    \ge0.
\]
By Lemma~\ref{lem:type-kernel}, this is exactly
$
    \sum_{\mathbf x\in\mathcal X(\mathbf m)}
    e_{\mathbf x}\mathsf K_{\mathbf x}(\mathbf y)
    \ge0.
$
Thus, the averaged solution gives a feasible solution to
$(\mathrm{TC\text{-}LP})$ with the same objective value.

Conversely, let $(e_{\mathbf x})$ be feasible for
$(\mathrm{TC\text{-}LP})$, and define $E_T:=e_{\operatorname{type}(T)}.$
The objective is preserved because
\[
    \sum_{T\subseteq\cD}E_T
    =
    \sum_{\mathbf x\in\mathcal X(\mathbf m)}
    N(\mathbf x)e_{\mathbf x}.
\]
The zero-word, distance, and non-negativity constraints are immediate. For any support $S$ of type $\mathbf y$, Lemma~\ref{lem:type-kernel} turns the ordered dual constraint into the $\mathbf y$-constraint of $(\mathrm{TC\text{-}LP})$. Hence, the lifted solution is feasible for the ordered LP with the same objective value.

The two optima are equal.
\end{proof}

\subsection{LP lower bound for repair bandwidth}
\label{subsec:lp-repair}

Let
\[
    \mathcal M_{d,\ell}
    :=
    \left\{
    \mathbf m=(m_1,\ldots,m_\ell)\in\mathbb Z_{\ge0}^{\ell}:
    \sum_{s=1}^{\ell}m_s=d
    \right\}.
\]
Here, $m_s$ denotes the number of helpers that send $s$ symbols. For $\mathbf m\in\mathcal M_{d,\ell}$, define
\[
    \beta(\mathbf m)=\sum_{s=1}^{\ell}s m_s,
    \quad
    \alpha(\mathbf m)=d\ell-\beta(\mathbf m).
\]
The $\cB$-profile has type $\mathbf m$. The nonzero part of the complementary $\cA$-profile has type
\[
    \mathbf m^A=(m_{\ell-1},m_{\ell-2},\ldots,m_1),
\]
where coordinates of $\cA$-dimension zero have been deleted.

We define $L_A(\mathbf m):=U_q^{\mathrm{typ}}(\mathbf m^A,r+1)$, and define $L_B(\mathbf m) := U_q^{\mathrm{typ}}(\mathbf m,d-r+1)$ if $d\ge r$.

\begin{definition}[LP-admissible profile]
\label{def:lp-admissible}
A type profile $\mathbf m\in\mathcal M_{d,\ell}$ is \emph{LP-admissible} if
\[
    L_A(\mathbf m)
    \ge
    \max\{q^{\alpha(\mathbf m)-(r-h)\ell},1\},
\]
and, when $d\ge r$, $L_B(\mathbf m)\ge q^{h\ell}. $
Let $\Gamma_{\mathrm{LP}}\subseteq\mathcal M_{d,\ell}$ be the set of LP-admissible profiles.
\end{definition}

\begin{theorem}[LP lower bound on repair bandwidth]
\label{thm:lp-repair-bound}
Every feasible repair scheme induces a profile in $\Gamma_{\mathrm{LP}}$. If $\Gamma_{\mathrm{LP}}\neq\emptyset$, then
\[
    \beta(h+d,k,\ell,\Fq,h,d)
    \ge
    \min_{\mathbf m\in\Gamma_{\mathrm{LP}}}\beta(\mathbf m).
\]
If $\Gamma_{\mathrm{LP}}=\emptyset$, then no linear repair scheme exists for the reduced full-helper instance with the prescribed parameters.
\end{theorem}

\begin{proof}
Consider any feasible repair scheme, and let $\mathbf m$ be the type profile of its bandwidth vector $(b_j:j\in\cD)$. Then
\[
    \beta=\beta(\mathbf m),
    \quad
    \alpha=\alpha(\mathbf m).
\]

By Theorem~\ref{thm:main-reduction}, the scheme induces a null-space code $\cA$ with
\[
    |\cA|\ge q^{\alpha(\mathbf m)-(r-h)\ell},
    \quad
    d_{\min}(\cA)\ge r+1.
\]
After deleting zero-dimensional coordinates, $\cA$ has type profile $\mathbf m^A$. Hence, the LP upper bound gives
$
    |\cA|\le L_A(\mathbf m).
$
Since every code has size at least $1$, the necessary condition is
\[
    L_A(\mathbf m)
    \ge
    \max\{q^{\alpha(\mathbf m)-(r-h)\ell},1\}.
\]

If $d\ge r$, Theorem~\ref{thm:main-reduction} also gives a row-space code $\cB$ of type profile $\mathbf m$ with
\[
    |\cB|=q^{h\ell},
    \quad
    d_{\min}(\cB)\ge d-r+1.
\]
Thus, the LP upper bound gives
$
    q^{h\ell}=|\cB|\le L_B(\mathbf m).
$
Therefore $\mathbf m\in\Gamma_{\mathrm{LP}}$.

Every feasible repair scheme induces an LP-admissible profile, so its bandwidth is at least
$
    \min_{\mathbf m\in\Gamma_{\mathrm{LP}}}\beta(\mathbf m).
$
Taking the minimum over all repair schemes gives the claimed lower bound. If $\Gamma_{\mathrm{LP}}=\emptyset$, no feasible repair profile can pass the necessary LP tests.
\end{proof}

For original parameters $(n,k,\ell,q,h,d)$, apply the theorem to the reduced full-helper instance from Lemma~\ref{lem:reduction_d}, whose redundancy is $r'=h+d-k$.

\begin{algorithm}[H]
\caption{Type-compressed LP lower bound}
\label{algo:computing-bound}
\begin{algorithmic}[1]
    \Require $d,r,h,\ell,q$, where $r$ is the redundancy of the
full-helper instance
    \Ensure $\beta_{\mathrm{LP}}$ or \textnormal{Infeasible}

    \State Construct $\mathcal M_{d,\ell}$ and sort its elements by
    $
        \beta(\mathbf m)=\sum_{s=1}^{\ell}s m_s.
   $

    \For{$\mathbf m\in\mathcal M_{d,\ell}$ in this order}
        \State $\alpha(\mathbf m)\gets d\ell-\beta(\mathbf m)$.
        \State $\mathbf m^A\gets(m_{\ell-1},m_{\ell-2},\ldots,m_1)$.
        \State $L_A\gets U_q^{\mathrm{typ}}(\mathbf m^A,r+1)$.

        \If{$L_A<\max\{q^{\alpha(\mathbf m)-(r-h)\ell},1\}$}
            \State \textbf{continue}
        \EndIf

        \If{$d\ge r$}
            \State $L_B\gets U_q^{\mathrm{typ}}(\mathbf m,d-r+1)$.
            \If{$L_B<q^{h\ell}$}
                \State \textbf{continue}
            \EndIf
        \EndIf

        \State \Return $\beta_{\mathrm{LP}}=\beta(\mathbf m)$.
    \EndFor

    \State \Return \textnormal{Infeasible}.
\end{algorithmic}
\end{algorithm}

\begin{corollary}
Algorithm~\ref{algo:computing-bound} returns $\min_{\mathbf m\in\Gamma_{\mathrm{LP}}}\beta(\mathbf m)$, or reports infeasibility if $\Gamma_{\mathrm{LP}}=\emptyset$.
\end{corollary}
\begin{proof}
The algorithm scans all type profiles in nondecreasing order of $\beta(\mathbf{m})$ and applies exactly the admissibility tests from Definition~\ref{def:lp-admissible}.
\end{proof}

\begin{theorem}[Polynomial-time LP bound for fixed $\ell$]
For every fixed sub-packetization $\ell$, Algorithm~\ref{algo:computing-bound} computes the LP lower bound in time
\[
    O_\ell\!\left(d^{\ell-1}T_{\mathrm{LP}}((d+1)^\ell)\right),
\]
where $T_{\mathrm{LP}}(N)$ is the time for solving an LP with $O(N)$ variables and $O(N)$ constraints.
\end{theorem}

\begin{proof}
For fixed $\ell$, Algorithm~\ref{algo:computing-bound} runs in polynomial time in $d$. Indeed,
\[
    |\mathcal M_{d,\ell}|
    =
    \binom{d+\ell-1}{\ell-1}
    =
    O_\ell(d^{\ell-1}).
\]
For a fixed profile $\mathbf m=(m_1,\ldots,m_\ell)$, the row-space LP for $\cB$ has
\[
    N_B(\mathbf m):=\prod_{s=1}^{\ell}(m_s+1)\le(d+1)^\ell
\]
support-type variables and the same order of constraints. The null-space LP for $\cA$ has
\[
    N_A(\mathbf m):=
    \prod_{t=1}^{\ell-1}(m_{\ell-t}+1)
    =
    \prod_{s=1}^{\ell-1}(m_s+1)
    \le(d+1)^{\ell-1}.
\]
Thus, the total running time is
$
    O_\ell\!\left(
        d^{\ell-1}\,
        T_{\mathrm{LP}}\bigl((d+1)^\ell\bigr)
    \right),
$
where $T_{\mathrm{LP}}(N)$ is the time for solving an LP with $O(N)$ variables and constraints. In particular, any polynomial-time LP solver gives a polynomial-time algorithm in $d$ for fixed $\ell$.

The uncompressed LP has one variable and one dual constraint for each subset of $\cD$, hence size $2^d$. The type compression is exact by
Proposition~\ref{prop:type-compression}.
\end{proof}

\subsection{Numerical examples}
\label{subsec:numerical}
Remark~\ref{rem:dual-viewpoint} explains that the closed-form Singleton-, Hamming-, and Plotkin-type bounds can be viewed as arising from particular choices in the dual LP. This subsection illustrates the resulting repair-bandwidth lower bounds
on representative finite-parameter instances.

In accordance with Lemma~\ref{lem:reduction_d}, we work in the full-helper setting and assume
\(n=h+d\). Since repair bandwidth is integral, all plotted lower bounds are rounded up to integers.

 We denote by \(\beta_{\mathrm{cut}}\) the centralized cut-set bound, by
\(\beta_{\mathrm{cf}}\) the best closed-form bound from Section~\ref{sec:bound}, and by \(\beta_{\mathrm{LP}}\) the type-compressed LP bound computed by Algorithm~\ref{algo:computing-bound}.
We also plot the Reed--Solomon repair bandwidth \(k\ell\) as a baseline.

\paragraph{Case 1: lower-bound comparison for \(h=1,r=3,\ell=2,q=5\).}

Figure~\ref{fig:num-q5-l2} compares \(\beta_{\mathrm{cut}}\), \(\beta_{\mathrm{cf}}\), \(\beta_{\mathrm{LP}}\), and the RS repair bandwidth \(k\ell\) as the repair degree \(d\) varies. In this range, the LP bound improves on the explicit closed-form bound and remains below the RS repair baseline.

\begin{figure}[ht!]
    \centering
    \includegraphics[width=0.5\linewidth]{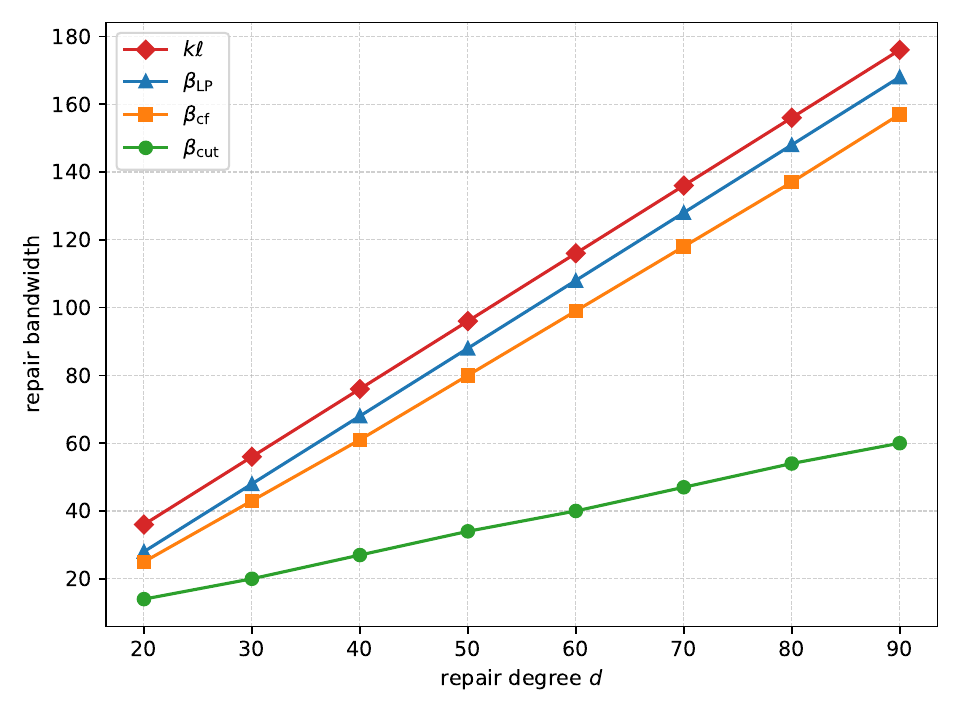}
    \caption{Comparison of lower bounds for \(h=1,r=3,\ell=2,q=5\). Here \(\beta_{\mathrm{cut}}\) is the cut-set bound, \(\beta_{\mathrm{cf}}\) is the best closed-form bound from Section~\ref{sec:bound}, \(\beta_{\mathrm{LP}}\) is the type-compressed LP bound, and \(k\ell\) is the repair bandwidth of traditional MDS repair.}
    \label{fig:num-q5-l2}
\end{figure}

\paragraph{Case 2: comparison with small-sub-packetization constructions.}

We next consider \(h=1,r=4,\ell=2,q=2^8\). The finite field of size \(q=2^8\) is widely used in storage-code implementations. In addition to the MDS repair baseline, we include LESS~\cite{cheng26} as a construction-side benchmark. For the corresponding parameters, the LESS repair bandwidth is
\[
    k+(\ell-1)\left\lceil \frac{n}{\ell+1}\right\rceil .
\]

Figure~\ref{fig:num-q256-l2-lower} compares the lower bounds \(\beta_{\mathrm{cut}}\), \(\beta_{\mathrm{cf}}\), and
\(\beta_{\mathrm{LP}}\). Figure~\ref{fig:num-q256-l2-constructions} compares the LP lower bound with the RS and LESS repair bandwidths.

\begin{figure}[ht!]
    \centering
    \begin{subfigure}{0.47\linewidth}
        \includegraphics[width=\linewidth]{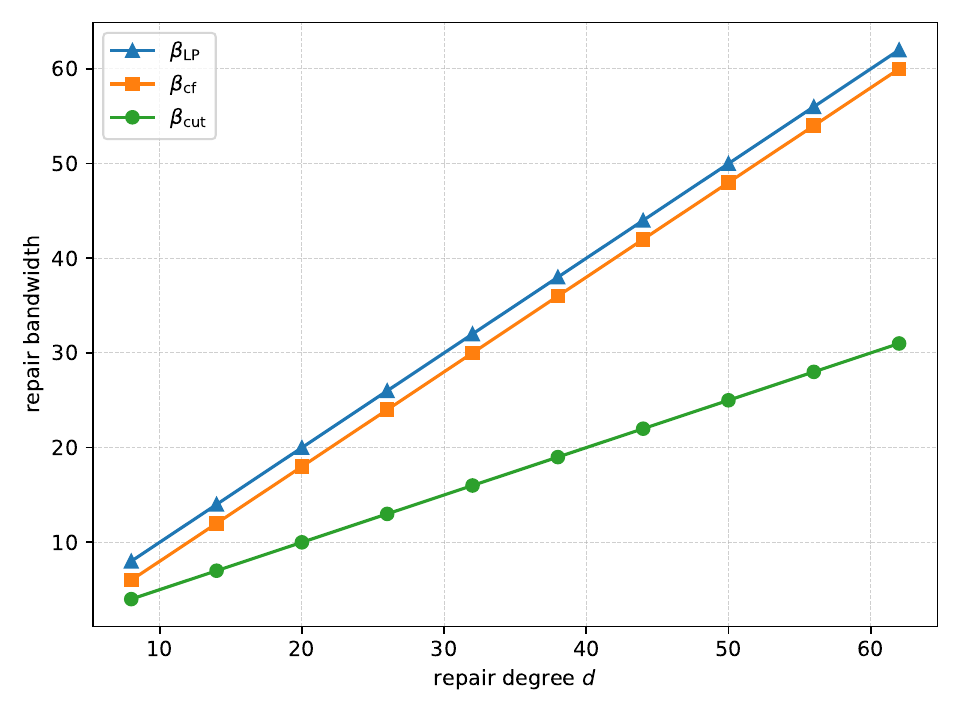}
        \caption{Lower-bound comparison.}
        \label{fig:num-q256-l2-lower}
    \end{subfigure}
    \hfill
    \begin{subfigure}{0.47\linewidth}
        \includegraphics[width=\linewidth]{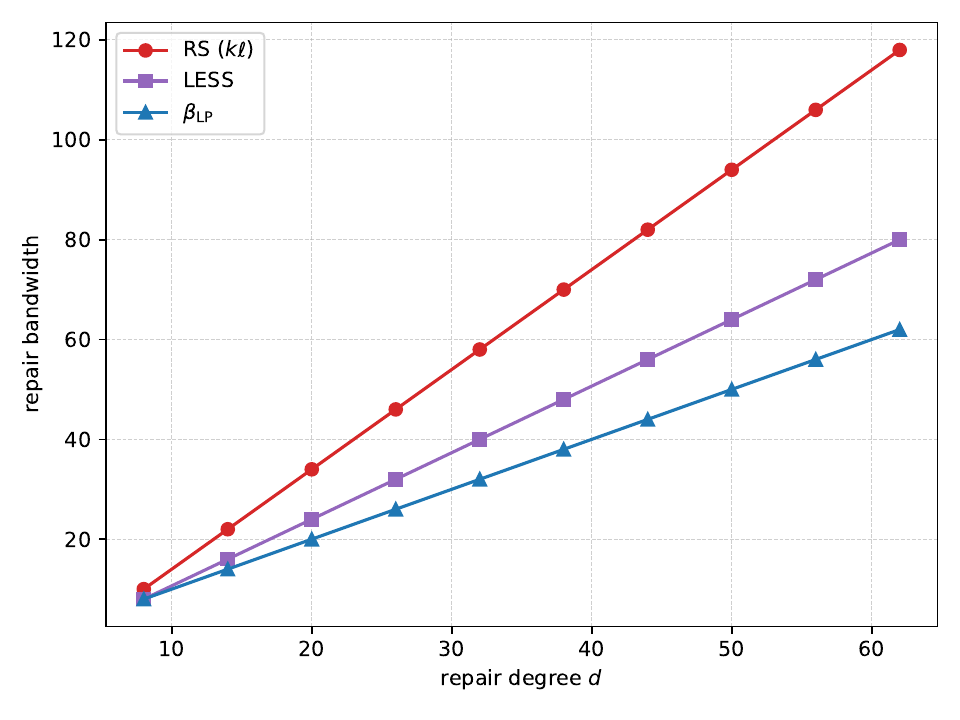}
        \caption{Comparison with RS and LESS bandwidths.}
        \label{fig:num-q256-l2-constructions}
    \end{subfigure}
    \caption{Numerical comparison for \(h=1,r=4,\ell=2,q=2^8\). In (a), \(\beta_{\mathrm{cut}}\) is the cut-set bound, \(\beta_{\mathrm{cf}}\) is the best closed-form bound from Section~\ref{sec:bound}, and \(\beta_{\mathrm{LP}}\) is the type-compressed LP bound. In (b), \(k\ell\) is the repair bandwidth of traditional MDS repair, \(k+(\ell-1)\left\lceil n/(\ell+1)\right\rceil\) is the repair bandwidth of LESS~\cite{cheng26}, and \(\beta_{\mathrm{LP}}\) is the LP lower bound.}
    \label{fig:num-q256-l2}
\end{figure}

These examples show that the type-compressed LP bound is computable for small and moderate finite parameters and can improve on the explicit closed-form bounds. They also indicate that a gap remains between the present lower bounds and these construction benchmarks.

\section{Conclusion}
\label{sec:conclusion}

We focus on lower bounds on the local minimum repair bandwidth for centralized linear repair of MDS array codes. The main step is a reduction from any linear repair scheme to auxiliary mixed-alphabet codes: a null-space code that is available for all repair degrees, and a row-space code that applies when $d\ge r$. Applying Singleton-, Hamming-, Plotkin-, and Griesmer-type bounds to these auxiliary codes yields several finite-parameter constraints on feasible bandwidth profiles.

We also gave a type-compressed LP algorithm for computing finite-parameter lower bounds. For fixed sub-packetization $\ell$, the compression is exact and reduces the support-enumerator LP from exponential size to polynomial size in $d$, yielding a polynomial-time procedure for evaluating the bound.

A natural direction for future work is to understand the relation between local and global repair optimality. The bounds developed here are local: they concern a fixed failure set and a fixed helper set. In contrast, global repair optimality requires one code to support low-bandwidth repair simultaneously for all repair patterns in a prescribed class. It would be interesting to characterize when local repair optima can be made mutually compatible to attain the global optimum, and, when they cannot, to quantify the resulting gap between local and global repair bandwidth.

\bibliographystyle{alpha}
\bibliography{ref}

\end{document}

%% file: def.tex
\usepackage{amsmath}
\usepackage{amssymb}
\usepackage{amsfonts}
\usepackage{amsthm}
\usepackage{array}
\usepackage{enumerate}

\usepackage{bm}
\usepackage{textcomp}

\usepackage{verbatim}
\usepackage{graphicx}

\usepackage{silence}
\usepackage{float}

\usepackage{soul}
\usepackage{color}
\usepackage{xcolor}
\usepackage{booktabs}

\usepackage{balance}

\usepackage[utf8]{inputenc}
\usepackage[T1]{fontenc}
\usepackage{url}
\usepackage{ifthen}
\usepackage{cite}
\usepackage{algorithm}
\usepackage{algorithmicx}
\usepackage{algpseudocode}

\allowdisplaybreaks

\newtheorem{theorem}{Theorem}
\newtheorem{lemma}[theorem]{Lemma}
\newtheorem{corollary}[theorem]{Corollary}

\newtheorem{proposition}[theorem]{Proposition}

\newtheorem{definition}{Definition}
\newtheorem{remark}{Remark}

\newtheorem{openproblem}{Open Problem}

\newlength{\spacelen}

\newcommand{\rk}[1]{\textnormal{rank}\left(#1\right)}

\newcommand{\Fq}{\mathbb{F}_q}

\newcommand{\bA}{\mathbb{A}}
\newcommand{\bB}{\mathbb{B}}

\newcommand{\bV}{\mathbb{V}}

\newcommand{\cA}{\mathcal{A}}
\newcommand{\cB}{\mathcal{B}}
\newcommand{\cC}{\mathcal{C}}
\newcommand{\cD}{\mathcal{D}}
\newcommand{\cE}{\mathcal{E}}

\newcommand{\cH}{\mathcal{H}}

\newcommand{\cR}{\mathcal{R}}

\newcommand{\cU}{\mathcal{U}}
\newcommand{\cV}{\mathcal{V}}

\DeclareMathOperator{\tB}{BW}